\documentclass[a4,11pt]{revtex4-1}
\usepackage{graphicx}
\usepackage{amsmath}
\usepackage{amsfonts,amssymb}
\usepackage[utf8]{inputenc}

\newcommand{\newc}{\newcommand}
\newc{\beq}{\begin{equation}}
\newc{\eeq}{\end{equation}}
\newc{\bea}{\begin{array}}
\newc{\eea}{\end{array}}
\newcommand{\ben}{\begin{eqnarray}}
\newcommand{\een}{\end{eqnarray}}
\newc{\ra}{\rightarrow}
\newc{\bfx}{{\bf x}}
\newc{\bfV}{{\bf V}}
\newc{\cO}{{\cal O}}
\newc{\bfv}{{\bf v}}
\newc{\bfu}{{\bf u}}
\newc{\bfp}{{\bf p}}
\newc{\ve}{{\varepsilon}}
\newc{\Psibar}{\overline\Psi}
\newc{\w}{{\bf w}}
\newc{\E}{{\mathbf{E}}}
\newc{\EE}{{\mathcal E}}
\newc{\bfn}{{\mathbf\nabla}}
\newc{\la}{{\cal L}}
\newc{\tla}{{\tilde{\cal L}}}
\newc{\bp}{{\bf p}}
\newc{\ho}{\hookrightarrow }
\newc{\bP}{{\bf P}}
\newc{\pd}{{\partial}}
\newc{\piv}{{\partial_4}}
\newc{\pv}{{\partial_5}}
\newc{\bJ}{{\bf J}}
\newc{\bze}{{\mathbf 0}}
\newc{\bK}{{\bf K}}
\newc{\tphi}{{\tilde\phi}}
\newc{\tF}{{\tilde F}}
\newc{\tD}{{\tilde D}}
\newc{\tJ}{{\tilde J}}
\newc{\tj}{{\tilde j}}
\newc{\bD}{{\bf D}}
\newc{\tvphi}{{\tilde\varphi}}
\newc{\trho}{{\tilde\rho}}
\newc{\ttheta}{{\tilde\theta}}
\newc{\tpsi}{{\tilde\psi}}
\newc{\tu}{{\tilde u}}
\newc{\cD}{{\cal D}}
\newc{\tPhi}{{\tilde\Phi}}
\newc{\tPsi}{{\tilde\Psi}}
\newc{\tA}{{\tilde A}}
\newc{\talpha}{{\tilde\alpha}}
\newc{\tbeta}{{\tilde\beta}}
\newc{\bA}{{\mathbf A}}
\newc{\bB}{{\bf B}}
\newc{\br}{{\bf r}}
\newc{\sig}{{\mathbf\sigma}}
\newc{\eg}{{\rm e.g.\ }}
\newc{\ie}{{\rm i.e.\ }}
\newcommand{\pslash}{\not{\hbox{\kern-2.3pt $p$}}}
\newcommand{\pdslash}{\not{\hbox{\kern-2pt $\partial$}}}
\newtheorem{theorem}{Theorem}[section]
\newtheorem{lemma}{Lemma}[section]
\newtheorem{proposition}{Proposition}[section]

\newenvironment{proof}[1][Proof]{\noindent\textbf{#1.} }{\ \rule{0.5em}{0.5em}}

\begin{document}

\title{INEQUALITIES FOR THE QUANTUM PRIVACY}

\author{M. A. S. TRINDADE}

\affiliation{Departamento de Ci\^{e}ncias Exatas e da Terra, Universidade do Estado da Bahia, Rodovia Alagoinhas/Salvador, BR 110, Km 03\\
Alagoinhas, Bahia, 48040-210, Brazil\\
}

\author{E. PINTO}

\affiliation{Departamento de Matem\'atica, Universidade Federal da Bahia, Campus Ondina\\
Salvador, Bahia 40210-340, Brazil\\
}

\begin{abstract}
In this work we investigate the asymptotic behavior related to the quantum privacy for multipartite systems. In this context, an inequality for quantum privacy was obtained by exploiting of quantum entropy properties. Subsequently, we derive a lower limit for the quantum privacy through the entanglement fidelity. In particular, we show that there is an interval where an increase in entanglement fidelity implies a decrease in quantum privacy.
\end{abstract}


\maketitle

\section{Introduction}
In quantum information science, the quantum cryptography exploits non-classical features of quantum systems to ensure security in the transmission of information \cite{Nielsen}. Classical encryption techniques have their security based on unproven computational difficulty of solving certain problems. The methods of quantum key distribution employ two parties to produce a random secret key known only to them. They are proven secure based on physical principles, without imposing any computational considerations. An interesting aspect is that quantum cryptography has given contributions to classical cryptography; amplification of privacy and classical bound information are examples of concepts in classical information whose discovery was inspired by quantum cryptography \cite{deutsch,Bennet2,Maurer1,Maurer2}. Besides, efforts have been made to an understanding of the relationship between private classical information and quantum information \cite{Wilde}. These facts make quantum cryptography an emerging field of great interest \cite{Gisin,Dongsu}.

The first ideas concerning quantum cryptography were proposed in 1970 by Wiesner, Bennett and Brassard in 1984 \cite{Wiesner1,Bennet1}. It is based on a combination of concepts of quantum physics and information theory. Substantial developments in quantum optics and optical fiber technology allowed its experimental realization \cite{Gisin,Bennet3}. Limitations on practical quantum cryptography has been analyzed by Brassard \textit{et al.} \cite{Gil}, showing that parametric down-conversion offers enhanced performance compared to its weak coherent pulse counterpart. The transmission loss limits in context of the continuous-variable quantum information were investigated by Namiki and Hirano \cite{Namiki} using coherent states and taking into account excess gaussian noise on quadrature distribution.

One of the problems in cryptography is to establish the limits of the techniques from the viewpoint of information theory. A theoretical formulation of quantum information channels of communication that allows a description of the limits in the context of information theory was proposed by Barnett and Phoenix \cite{Bar}. An important discovery in the theory of quantum information is the Holevo limit \cite{Holevo}. Given a sender (Alice) and receiver (Bob), this limit says that no matter how Bob perform their measurements, the mutual information between Alice and Bob can not exceed an amount $\chi$, called Holevo limit. In this scenario, Schumacher \cite{West} presented the definition of optimal guaranteed privacy for a quantum channel as
\begin{equation}
P=sup[H_{Bob:Alice}-H_{Eve:Alice}],
\end{equation}
where the supremum is taken over all strategies that Alice and Bob may employ to use the channel, $H$ is the mutual information and Eve is the eavesdropper. In addition, he derived based on the Holevo limit, a lower limit for privacy in terms of coherent information $I$:
\begin{equation}
P\geq{S(\rho^{B})-S(\rho^{E})=I(\rho,\varepsilon}), \label{P2}
\end{equation}
where $S$ is the von Neumann entropy, the indices $B$ and $E$ refer to Bob and Eve, respectively, and $\varepsilon$ characterizes the noise channel.

In this article we present some general inequalities for quantum privacy in the context of quantum cryptography. The paper is unfolded in the following sequence of presentation. In the Section \ref{sec:QP1} we developed some general asymptotic results about of quantum privacy. Section \ref{sec:QP2} contains a relationship between the quantum privacy and the quantum fidelity. Sec. \ref{sec:conclu}, we present the conclusions and some perspectives.

\section{Asymptotic Behavior}\label{sec:QP1}

Asymptotic results are becoming more common in the theory of quantum information. The theory of large deviations has been developed in a quantum scenario and asymptotic estimates for the probability of events which are useful for identification in noisy channels were obtained \cite{Ld}. Bae and Ac\'in \cite{Bae} proved that quantum cloning becomes equivalent to state estimation in the asymptotic regime where the numbers of clones tends to infinity. The asymptotic context is  present in the quantum information spectrum approach \cite{Bowen,Hay} and it can provide, for example, a general expression for the classical capacity of arbitrary quantum channels. It extends the scheme high probability events to high probability subspaces of states defined on a Hilbert space. An asymptotic theory of quantum inference \cite{Hay1} has also been developed. Recent developments in asymptotic quantum hypothesis testing have been obtained \cite{Nus1,Nus2,Aud}. In particular, the multiple hypothesis testing problem for symmetric quantum state discrimination was addresses and upper bounds on asymptotic error exponents were derived \cite{Aud}. If we want to fully know the density operator $\rho$, we need an infinite ensemble of quantum systems prepared in the same state, which is impossible in pratice \cite{Hay1}. The impossibility of perfect state estimation is a fundamental result \cite{Bae} in quantum mechanics and for the large sample case, we can apply the asymptotic theory. Asymptotic security of quantum key distribution under collective attacks was analyzed by Zhao \emph{et al.} \cite{zhao}. A lower bound to the secret key rate for an binary modulated continuous-variable quantum key distribution was evaluated. In reference \cite{Weh} Wehner and Winter discuss some open problems of interest for the foundation of the security of quantum cryptographic protocols assuming an asymptotic view point of large dimension. The problem of finding the optimal reversing channel on n-qudits ensembles, usefull to quantum key distribution, was solved using quantum local asymptotic normality \cite{Bowles} (an extension of an important result from classical asymptotic statistics). The optimal strategy, in order to reverse the action of an arbitrary channel acting on an ensemble of n qudits is optimally reverse the associated Gaussian channel and map the output back onto m qudits. More recently, in a study of long-distance quantum communication and cryptography beyond the use of entanglement distillation, an asymptotic version of the distinguishability bound has been derived by Bauml \emph{et al.} \cite{Bauml}.

In this scenario, the following lemma provides some asymptotic results about the quantum entropy that enables us to derive a general inequality for quantum privacy. A key ingredient is the semicontinuity of quantum entropy \cite{We,Oh}.
\begin{lemma}
Let $(\rho_{n}^{B})$ and $(\rho_{n}^{E})$ be two sequences of density operators satisfying the following conditions:\\

(i) $tr|\rho_{n}^{B}-\rho^{B}|\rightarrow 0$; \\

(ii) $\limsup_{n \rightarrow \infty}S(\rho_{n}^{B})\leq S(\rho_{\ast}^{B})$ and $\limsup_{n \rightarrow \infty}S(\rho_{n}^{E})\leq S(\rho^{E})$. \\
Then
\begin{equation}
\liminf_{n \rightarrow \infty}[S(\rho_{n}^{B})-S(\rho_{n}^{E})]\geq S(\rho^{B})-S(\rho^{E})
\end{equation}
and
\begin{equation}
\liminf_{n \rightarrow \infty}[-S(\rho_{n}^{B})-S(\rho_{n}^{E})]\geq -S(\rho_{\ast}^{B})-S(\rho^{E}).
\end{equation}
\end{lemma}

\begin{proof}
(i) Since $tr|\rho_{n}^{B}-\rho^{B}|\rightarrow 0$, by use the semicontinuity of quantum entropy \cite{We,Oh}, we have that
\begin{equation*}
\liminf_{n \rightarrow \infty}S(\rho_{n}^{B})\geq S(\rho^{B})
\end{equation*}
Consequently
\begin{eqnarray*}
\liminf_{n \rightarrow \infty}[S(\rho_{n}^{B})-S(\rho_{n}^{E})]&\geq& \liminf_{n \rightarrow \infty}S(\rho_{n}^{B})+\liminf_{n \rightarrow \infty}[-S(\rho_{n}^{E})] \nonumber \\
&=&\liminf_{n \rightarrow \infty}S(\rho_{n}^{B})-\limsup_{n \rightarrow \infty}S(\rho_{n}^{E}) \nonumber \\
&\geq& S(\rho^{B})-S(\rho^{E}).
\end{eqnarray*}
Analogously,
\begin{eqnarray*}
\liminf_{n \rightarrow \infty}[-S(\rho_{n}^{B})-S(\rho_{n}^{E})]&\geq& \liminf_{n \rightarrow \infty}[-S(\rho_{n}^{B})]+\liminf_{n \rightarrow \infty}[-S(\rho_{n}^{E})]  \\
&=&-\limsup_{n \rightarrow \infty}S(\rho_{n}^{B})-\limsup_{n \rightarrow \infty}S(\rho_{n}^{E}) \\
&\geq& -S(\rho_{\ast}^{B})-S(\rho^{E}).
\end{eqnarray*}
\end{proof}

With the previous lemma, we can deduce an inequality for quantum privacy. Similarly to the reference \cite{West}, we will consider the states initially prepared by Alice as pure states and the environment also starts in a pure state.
\begin{theorem}
Suppose Alice sends states $\rho_{k}^{A_{1},A_{2},\ldots,A_{N}}$, where $k=0,1,2,\ldots$ denotes the possible states with probability $p_{k}$ in a system of dimension $d$. Bob receives states $\rho_{k}^{B}=\rho_{k}^{B_{1},B_{2},\ldots,B_{N}}=\varepsilon(\rho_{k}^{A_{1},A_{2},\ldots,A_{N}})$ due the interference of Eve. Besides, let $(\rho_{n}^{B_{i}})$ and $(\rho_{n}^{E_{i}})$ two sequences of density operators satisfying the conditions of previous lemma. Thus,
\begin{equation}
\sup_{n}P \geq I(\rho^{(1)},\varepsilon)-I'(\rho^{(2)},\varepsilon)- \cdots-I'(\rho^{(n)},\varepsilon),
\end{equation}
where
\begin{equation}
I'(\rho^{(i)},\varepsilon)=-S(\rho_{\ast}^{B_{i}})-S(\rho^{E_{i}}),
\end{equation}
and $\rho^{B_{i}}$ (or $\rho^{E_{i}})$ is the reduced state of $\rho^{B_{1},B_{2},\ldots,B_{N}}$ (or $\rho^{E_{1},E_{2},\ldots,E_{N}}$) on the system $B_{i}$ (or $E_{i}$).
\end{theorem}

\begin{proof}
We have that
\begin{eqnarray*}
P_{n} &\geq& S(\rho_{n}^{B_{1},B_{2},\ldots,B_{N}})-S(\rho_{n}^{E_{1},E_{2},\ldots,E_{N}}) \\
&\geq& S(\rho_{n}^{B_{1}})-S(\rho_{n}^{B_{2}})- \cdots -S(\rho_{n}^{B_{N}})-S(\rho_{n}^{E_{1}})- \cdots -S(\rho_{n}^{E_{N}}) \\
&=&S(\rho_{n}^{B_{1}})-S(\rho_{n}^{E_{1}})-S(\rho_{n}^{B_{2}})-S(\rho_{n}^{E_{2}})- \cdots -S(\rho_{n}^{B_{N}})-S(\rho_{n}^{E_{N}}) \\
\end{eqnarray*}
by Araki-Lieb triangle inequality \cite{araki}. Using the previous lemma
\begin{eqnarray*}
\sup_{n}P_{n} &\geq& \liminf_{n \rightarrow \infty}[S(\rho_{n}^{B_{1}})-S(\rho_{n}^{E_{1}})]+\liminf_{n \rightarrow \infty}[-S(\rho_{n}^{B_{2}})-S(\rho_{n}^{E_{2}})] \\
&+& \cdots + \liminf_{n \rightarrow \infty}[-S(\rho_{n}^{B_{N}})-S(\rho_{n}^{E_{N}})] \\
&\geq& S(\rho^{B_{1}})-S(\rho^{E_{1}})-S(\rho_{\ast}^{B_{2}})-S(\rho^{E_{2}})- \cdots-S(\rho_{\ast}^{B_{N}})-S(\rho^{E_{N}}) \\
&=&I(\rho^{(1)},\varepsilon)-I'(\rho^{(2)},\varepsilon)- \cdots -I'(\rho^{(N)},\varepsilon).
\end{eqnarray*}
\end{proof}

Importantly, this result generalizes the equation (\ref{P2}) for multipartite systems in an asymptotic context. In particular, for independent channels ($\varepsilon=\varepsilon_{1}\otimes\varepsilon_{2}\cdots \otimes\varepsilon_{N}$) and product states ($\rho=\rho_{1}\otimes\rho_{2}\cdots \otimes\rho_{N}$), we have
\begin{eqnarray}
\sup_{n}P_{n} &\geq& I(\rho^{(1)},\varepsilon_{1})+I(\rho^{(2)},\varepsilon_{2})+ \cdots +I(\rho^{(N)},\varepsilon_{N}) \nonumber \\
              &=&\sum_{i=1}^{N}I(\rho^{(i)}, \varepsilon_{i}), \nonumber
\end{eqnarray}
applying the additivity property of quantum entropy for product states.

\section{Quantum Privacy and Quantum Fidelity}\label{sec:QP2}

The measures of distinguishability of the quantum states are useful for analysis of the security in the transmission of quantum information. Biham and Mor presented in reference \cite{Biham} new limits on such measures and used these limits to prove security against a large class of attacks quantum key distribution. A distance measure widely used in information theory is quantum entanglement fidelity introduced by Schumacher in 1996 \cite{Sch}. It gives the amount of entanglement preserved in a quantum operation $\varepsilon$. Nielsen showed \cite{N2} that entanglement fidelity is the important quantity to maximize in schemes for quantum error correction. It determines how well the system state under error correction is maintained and how the entanglement with the auxiliary systems is preserved \cite{Beny}. In the following proposition we investigate a relationship between the entanglement fidelity and the optimal guaranteed privacy for a quantum channel.

\begin{proposition}
Suppose Alice sends states $\rho_{k}^{A}$ where k = 0,1,2 \ldots denotes the possible states with probability $p_{k}$, in a system of dimension $d$. Bob then receives  states $\rho_{k}^{B}=\varepsilon(\rho_{k}^{A})$ due to the interference of Eve. Let $p=infP$ on $\rho$. Then, since that $\varepsilon$ is unitary operation , $p$ has a maximum at $F = (d^2-1)/d^2$, where $F$ is the entanglement fidelity.
\end{proposition}

\begin{proof}
The quantum Fano inequality states that \cite{Sch}
\begin{equation}
S(\rho, \varepsilon)\leq{H(F(\rho,\varepsilon))+(1-F(\rho,\varepsilon))log(d^2-1)},
\end{equation}
where $H$ is the binary Shannon entropy. Combining this inequality with the expression (\ref{P2}) and identifying $ S(\rho, \varepsilon)=S(\rho^{E})$ and $S(\varepsilon(\rho))=S(\rho^{B})$, we find
\begin{equation}
P\geq{S(\rho^{B})-FlogF-(1-F)log(1-F)+(F-1)log(d^2-1)}. \label{P4}
\end{equation}
Since $p=infP$, we have a maximum value p at $F = (d^2-1)/d^2$.
\end{proof}

In the figure (\ref{fig1}) we have $S=6$ e $d = 2,3,4$ and $8$, respectively. The maximum occurs at $F = 0.75$, $F=0.89$, $F = 0.94$ and  $F = 0.98$, respectively. This behavior is exhibited due to the presence of the binary Shannon entropy in the expression (\ref{P4}). We can also observe that increasing in the dimension of the system, the fidelity to the maximum value of $p$ becomes close to 1. This result indicates that an increase in the preservation of entanglement implies, since the fidelity is within a certain interval (small), a decrease in privacy. Note that this interval becomes smaller by increasing the size of the system, as shown in graphs.
\begin{figure}[th]
\centering
\includegraphics[scale=0.298]{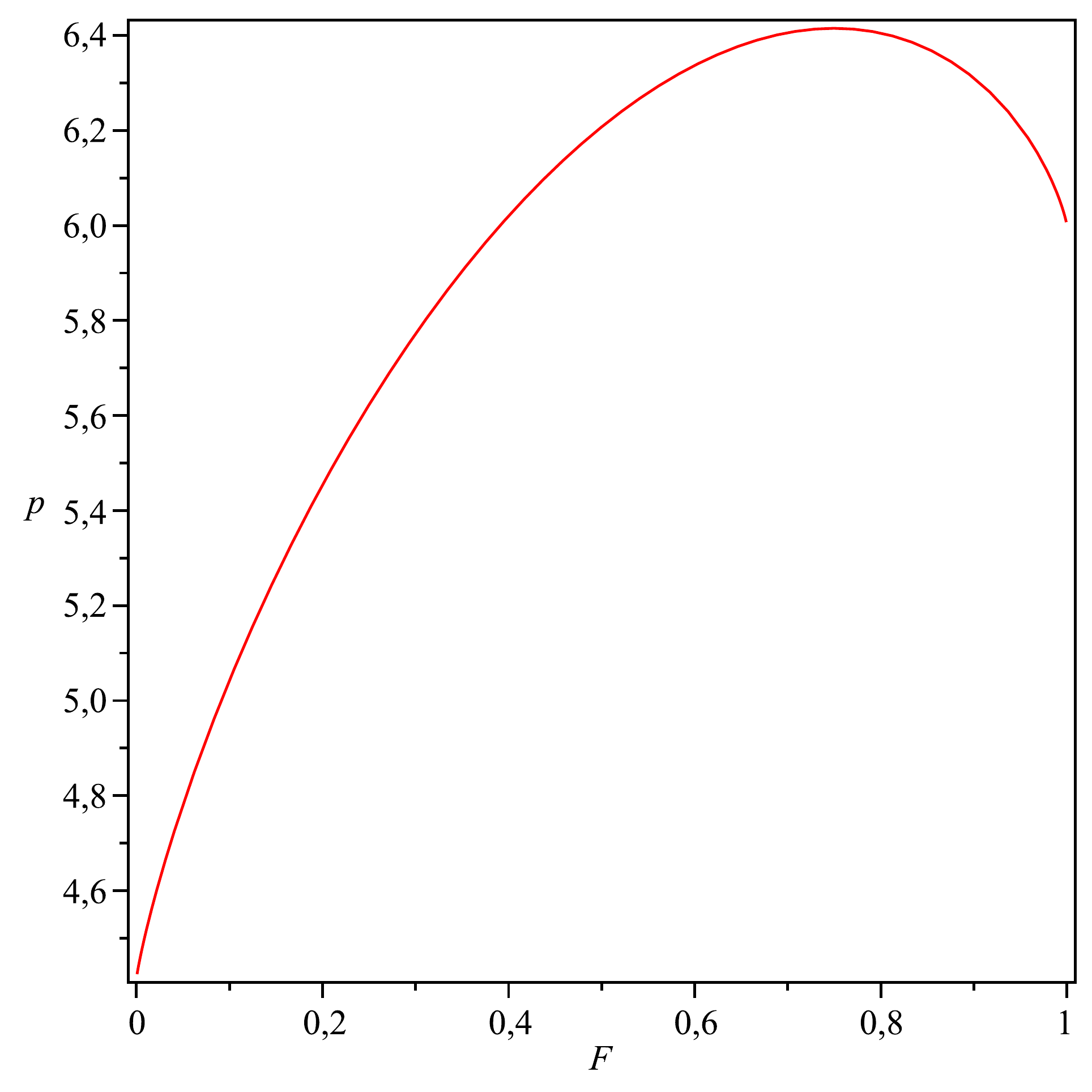}
\includegraphics[scale=0.298]{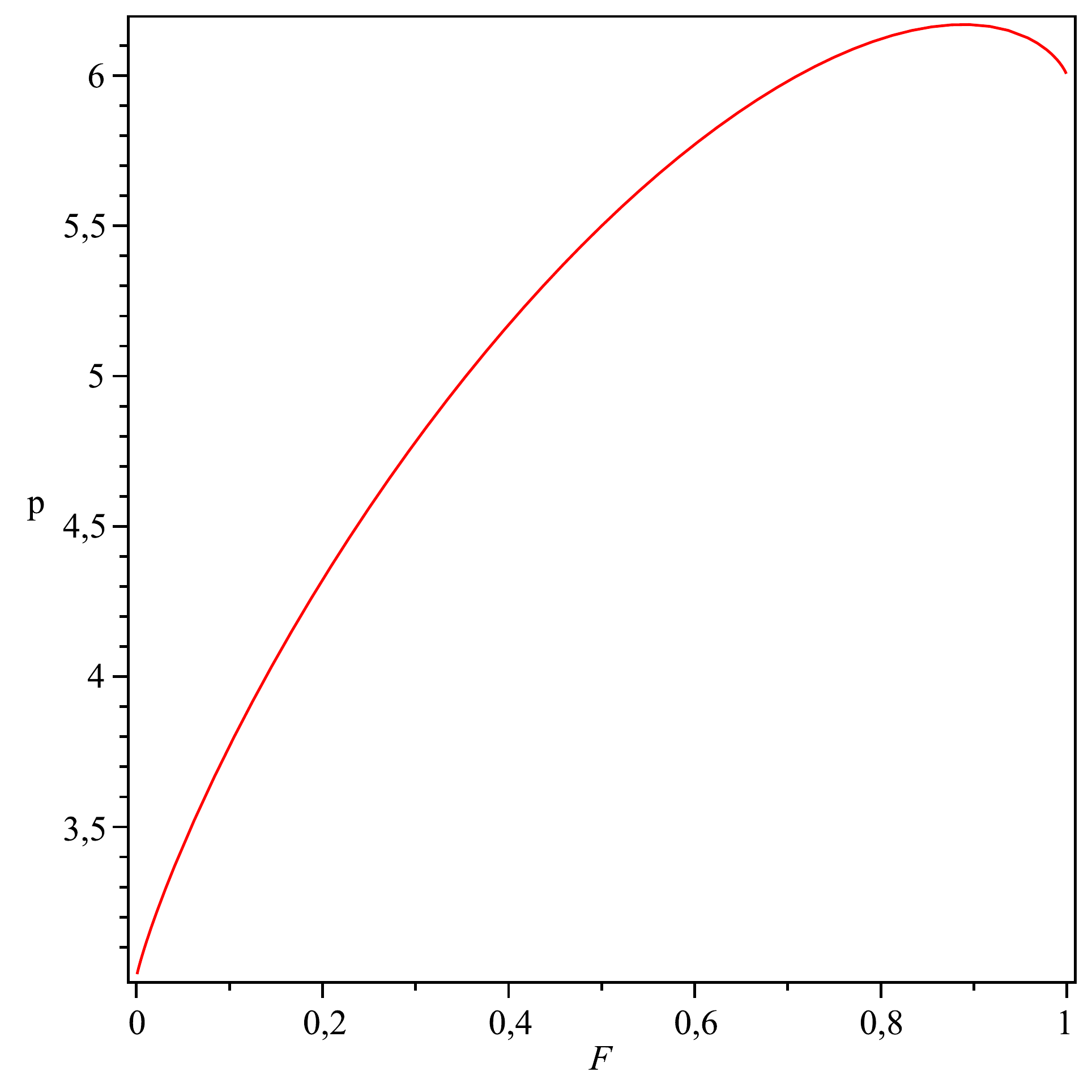}
\includegraphics[scale=0.298]{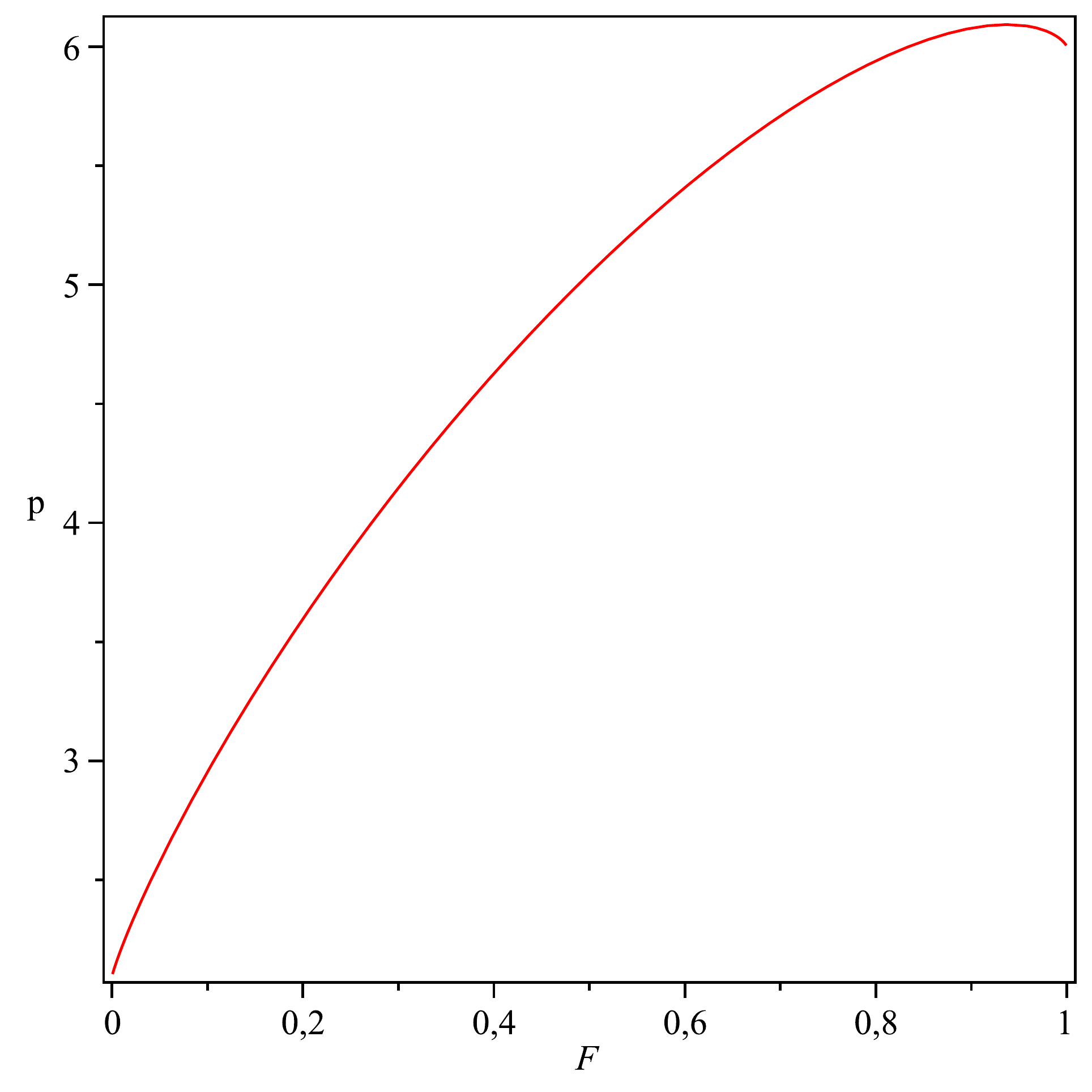}
\includegraphics[scale=0.298]{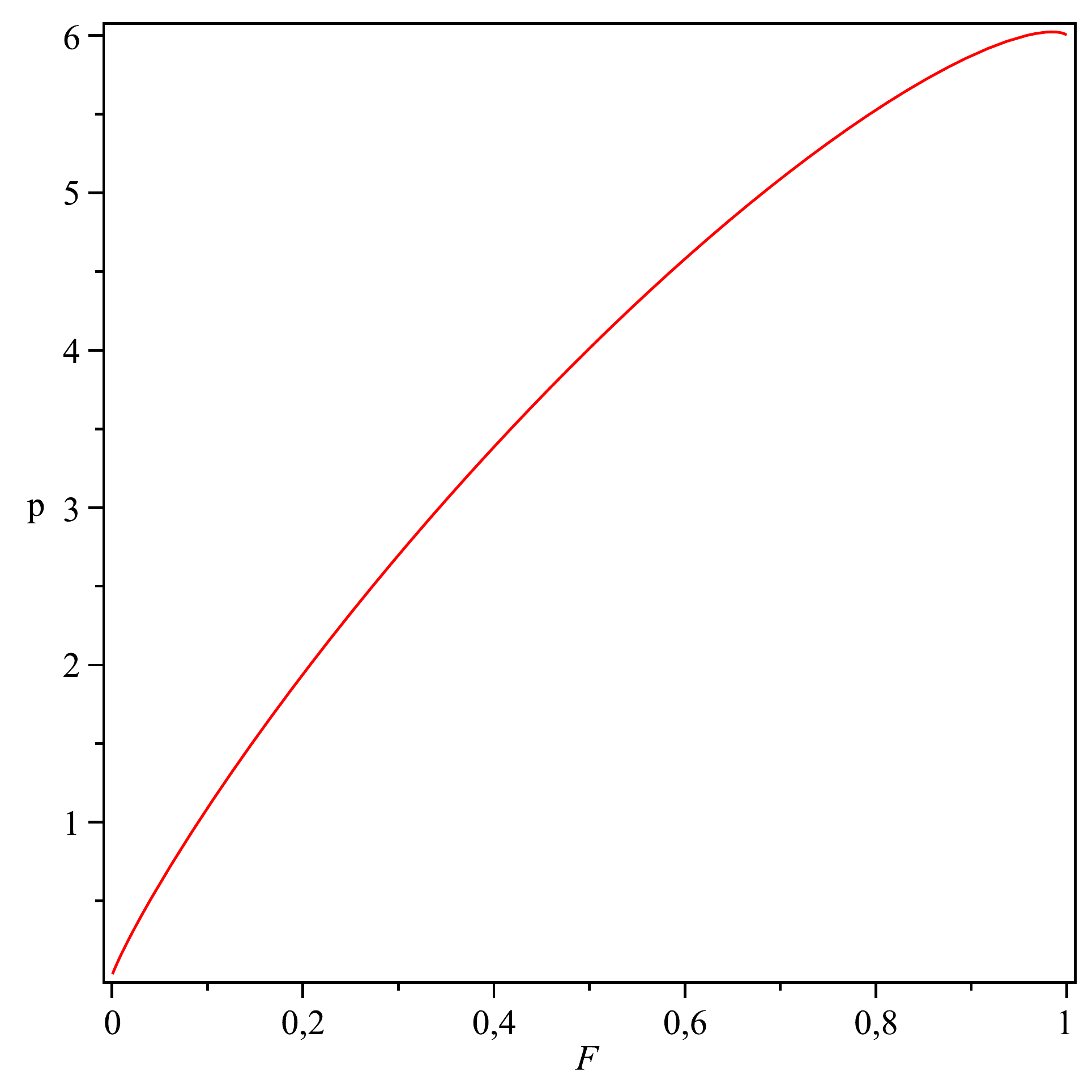}
\caption{$p$ as a function of fidelity for dimension $S(\rho^{B})=6$ and $d=2,3,4,8$, respectively.}
\label{fig1}
\end{figure}
\newpage
\section{Conclusions}\label{sec:conclu}
In this paper, we present a generalization of the inequality (\ref{P2}) to multipartite systems in an asymptotic scenario. The inequalities deduced herein establish asymptotic bounds for quantum privacy using a large number of quantum systems, which the key ingredient was the semicontinuity of quantum entropy. Next, we derived a relationship between the entanglement fidelity ant quantum privacy. We believe that this result is important because it indicates that an increase in the preservation of entanglement implies, since the fidelity is within a certain interval (small), a decrease in privacy, revealing an interesting aspect of these concepts in quantum information theory. We emphasize that result presented is not specific of any protocol and it is independent of the degree of entanglement. An investigation to obtain limits for privacy using the degree of entanglement are in progress.


\begin{thebibliography}{0}
\bibitem{Nielsen} M. A. Nielsen and I. L. Chuang, \textit{Quantum Computation and Quantum Information} (Cambridge University Press, Cambridge, 2000).

\bibitem{deutsch} David Deutsch \emph{et al.}, \textit{Phys. Rev. Lett.} \textbf{77} (1998) 2818.

\bibitem{Bennet2} C. H. Bennet, G. Brassard and J. M. Robert, \textit{SIAM J. Comp}. \textbf{17}  (1988) 210-229.

\bibitem{Maurer1} U. M. Maurer, \textit{IEEE Transactions on Information Theory} \textbf{39} (1993) 733.

\bibitem{Maurer2} U. M. Maurer and S. Wolf, \textit{IEEE Trans. Inf. Theory} \textbf{45} (1999) 499-514.

\bibitem{Wilde} M. M. Wilde and M.-H. Hsieh, \textit{Quantum Inf. Process.} \textbf{11} (2012) 1465-1501.

\bibitem{Gisin} N. Gisin, G. Ribordy, W. Tittel and H. Zbinden,  \textit{Rev. Mod. Phys.} \textbf{28} 15 (2014) 1450121.

\bibitem{Dongsu} Dongsu Shen, Wenping Ma and Meiling Wang, \textit{Mod. Phys. Lett. B} \textbf{28} (2014) 15 1450121.

\bibitem{Wiesner1} S. Wiesner, \textit{ACM SIGACT News} \textbf{15} (1993) 1.

\bibitem{Bennet1} C. H. Bennet and G. Brassard, \textit{Proceedings of IEEE International Conference on Computers, Systems  Signal
    Processing, Bangalore, India}(IEEE, New York, 1984) 175.

\bibitem{Bennet3} C. H. Bennet, F. Bessete, G. Brassard, L. Salvail and J. Smolin,  \textit{J. Cryptology} \textbf{15} (1992) 3.

\bibitem{Gil} G. Brassard, N. Lütkenhaus, T. Mor and B. C. Sanders, \textit{Phys. Rev. Lett.} \textbf{85} (2000) 6.

\bibitem{Namiki} R. Namiki and T. Hirano, \textit{Phys. Rev. Lett.} \textbf{92} (2004) 11.

\bibitem{Bar} S. M. Barnett and S.J.D. Phoenix, \textit{Phys. Rev. A} $\textbf{48}$ (1993) 1.

\bibitem{Holevo} A. S. Holevo, \textit{Probl. Peredachi Inf.}. \textbf{9} (1973) 3  [\textit{Probl. Inf. Trans.} (Engl. Trans.) \textbf{9} (1973) 110].

\bibitem{West} B. Schumacher and M. D. Westmoreland, \textit{Phys. Rev. Lett.} \textbf{80} (1998) 25.

\bibitem{Ld} R. Ahlswede, V. M. Blinovsky, \textit{Probl.  Inf. Trans.} \textbf{39} (2003) 4 373-379.


\bibitem{Bae} J. Bae and A. Ac\'in, \textit{Phys. Rev. Lett.} \textbf{97} (2006) 030402.

\bibitem{Bowen} G. Bowen and N. Datta, \textit{e-print arXiv:quant-ph/0610003}.

\bibitem{Hay} M. Hayashi and H. Nagaoka, \textit{IEEE Trans. Inform. Theory} \textbf{49} (2003) 1753.

\bibitem{Hay1} M. Hayashi, \textit{Asymptotic Theory of Quantum Statistical Inference: selected papers}  (World Scientific, Hackensack NJ, 2005).

\bibitem{Nus1} M. Nussbaum and A. Szkola, \textit{Ann. Statist.} \textbf{39} (2011) 6 3211-3233.

\bibitem{Nus2} M. Nussbaum, \textit{Proceedings of The First International Workshop on Entangled Coherent States and Its Application to Quantum Information Science, T. S. Usuda, K. Kato, Eds. Japan}(Tamagawa University, Tokyo, 2013) 77-81.

\bibitem{Aud} K. M. R. Audenaert, M. Mosonyi, \textit{J. Math. Phys.} \textbf{55} (2014) 55 102201.


\bibitem{zhao} Yi-Bo Zhao, M. Heid, J. Rigas and N. Lutkenhaus, \textit{Phys. Rev. A} \textbf{79} (2009) 012307.


\bibitem{Weh} S. Wehner and A. Winter, \textit{New J. Phys.} \textbf{12} (2010) 12 025009.


\bibitem{Bowles} P. Bowles, M. Guta and G. Adesso \textit{New J. Phys.} \textbf{14} (2012) 113041.

\bibitem{Bauml} S. Bauml, M. Christandly, K. Horodecki and A. Winter, \textit{e-print arXiv:quant-ph/14025927} (2014).


\bibitem{We} A. Wehrl, \textit{Rev. Mod. Phys.} \textbf{50} (1978) 221.

\bibitem{Oh} M. Ohya and N. Watanabe, \textit{Entropy} \textbf{12} (2010) 1194.

\bibitem{araki} H. Araki and E. Lieb, \textit{Commun. Math. Phys.} \textbf{18}(1970) 160-170.

\bibitem{Biham} E. Bihama and T. Mor, \textit{Phys. Rev. Lett.} \textbf{79} (1997) 20.

\bibitem{Sch} B. W. Schumacher, \textit{Phys. Rev. A} \textbf{54} (1996) 2614.

\bibitem{N2} M. A. Nielsen, \textit{e-print arXiv: 9606012}.

\bibitem{Beny} C. Bény and O. Oreshkov, \textit{ Phys. Rev. Lett.} \textbf{104} (2010) 120501.
\end{thebibliography}
\end{document}